\documentclass[conference]{IEEEtran}
\usepackage{setspace,verbatim,amsfonts,graphicx,amsmath,amsbsy,amssymb,citesort,epsfig,epsf,url}
\newtheorem{theorem}{Theorem}[section]

\newtheorem{lemma}[theorem]{Lemma}

\newtheorem{definition}{Definition}[section]

{}

\newcommand{\Zb}{{\mathbf Z}}

\newcommand{\ab}{{\mathbf a}}

\newcommand{\gb}{{\mathbf g}}

\newcommand{\pb}{{\mathbf p}}
\newcommand{\qb}{{\mathbf q}}
\newcommand{\rb}{{\mathbf r}}

\newcommand{\xb}{{\mathbf x}}

\newcommand{\zb}{{\mathbf z}}

\ifCLASSINFOpdf
\else
\fi
\begin{document}
%
\title{Compressive MUSIC with Optimized Partial Support   for Joint Sparse Recovery}

\author{\IEEEauthorblockN{Jong Min Kim,
Ok Kyun Lee and
Jong Chul Ye}
\IEEEauthorblockA{Department of Bio and Brain Engineering, KAIST,
Guseongdong, Daejeon, 305-701, Korea\\ Email: franzkim@gmail.com, jong.ye@kaist.ac.kr  }}


\maketitle

\begin{abstract}

The multiple measurement vector (MMV) problem addresses the identification
of unknown input vectors that share common sparse support.  
The MMV problem has been traditionally addressed either by
sensor array signal processing or compressive sensing. However, recent breakthroughs in this area such as compressive MUSIC (CS-MUSIC) or subspace-augumented MUSIC (SA-MUSIC)
 optimally combine the compressive sensing (CS) and array signal processing   such that $k-r$ supports are first found by CS  and the remaining  $r$ supports are determined by a generalized MUSIC criterion, where $k$ and $r$ denote the sparsity and the number of independent snapshots, respectively. Even though such a hybrid approach significantly outperforms the conventional algorithms, 
its performance  heavily depends on the correct identification of  $k-r$ partial support by the compressive sensing step, which often deteriorates the overall performance.
The main contribution of this paper is, therefore, to show that as long as $k-r+1$ correct supports are included in any $k$-sparse CS solution, the optimal $k-r$ partial support can be found using a subspace fitting criterion, significantly improving the overall performance of CS-MUSIC. Furthermore, unlike the single measurement CS counterpart that requires infinite SNR for a perfect support recovery, 
we can derive an information theoretic sufficient condition for the perfect recovery using CS-MUSIC under a {\em finite} SNR scenario.
\end{abstract}
\IEEEpeerreviewmaketitle

\section{Introduction}

One of important areas of compressed sensing research is the so-called
multiple measurement vector problem (MMV)
\cite{chen2006trs,Kim2010CMUSIC,LeeBresler2010,Feng97,KLY2011opt-CSMUSIC,DaviesEldar2010,KCJBY2011,malioutov2005ssr}. The MMV
problem addresses the recovery of a set of sparse signal vectors
that share common non-zero support. 
In MMV, thanks to the common sparse support, it is quite predictable
that the number of recoverable sparsity levels may increase with the
increasing number of measurement vectors. 
However,  the performance
 of the existing MMV compressive sensing algorithms  are not generally satisfactory even for
a noiseless case when a finite number of snapshots are available.

A recent breakthrough in this area has created a new class of algorithms such as compressive MUSIC (CS-MUSIC) proposed by our group \cite{Kim2010CMUSIC} or subspace-augumented MUSIC (SA-MUSIC)  proposed independently \cite{LeeBresler2010}.  Specifically, when the number of targets is $k$, and $r$ independent snapshots are available,  compressive MUSIC  finds $k-r$ targets using a compressive sensing algorithm such as S-OMP or $p$-thresholding, and the remaining $r$ targets are recovered using a generalized MUSIC criterion \cite{Kim2010CMUSIC}.
This hybridization significantly improves the performance of estimating jointly sparse signals and achieves the $l_0$ sparse recovery bound using a finite number of snapshots. 
Furthermore, even if the sparsity level is not known
{\em a priori}, compressive MUSIC can accurately estimate the
sparsity level using the generalized MUSIC criterion. 
In spite of its success, one of the main shortcomings of CS-MUSIC or SA-MUSIC is that the overall performance is heavily dependent upon the success of the first $k-r$ support estimation. This is especially problematic when the measurement is so noisy or the RIP condition for the sensing matrix is so bad that the greedy $k-r$ update step may produce incorrect support estimate.

One of the main contributions of this paper is, therefore, to relax  this  stringent requirement. In particular, the new algorithm requires that 
 $k-r+1$ supports (not in sequential order) out of $k$ support estimation is correct rather than $k-r$ consecutive support estimate are correct.  The location of the unknown $k-r$ true support can be then readily estimated using a {subspace fitting} criterion.  Such optimized partial support estimates can significantly improve the accuracy of the generalized MUSIC step, hence overall performance of compressive MUSIC.

The paradigm shift from early termination of CS algorithm after $k-r$ step to  selecting the correct $k-r$ supports out of $k$-sparse solution by any  CS algorithm is much more significant  and fundamental than just regarded as algorithmic improvement. In particular, by converting the problem as a partial support recovery problem, we can adapt rich information theoretical analysis tools that have been developed for single measurement vector CS (SMV-CS) \cite{ReevesGastpar2008} .
In particular, we can  derive an information theoretic sufficient condition for the perfect recovery of CS-MUSIC under a {\em finite} SNR scenario, which was considered not feasible in the SMV-CS \cite{ReevesGastpar2008} .

\section{Problem Formulation and Mathematical Preliminaries}
\label{sec:formulation}

Throughout the paper, $\xb^i$ and $\xb_j$ correspond to the $i$-th
row and the $j$-th column of matrix $X$, respectively. When $S$ is an
index set, $X^S$, $A_S$ corresponds to a submatrix collecting
corresponding rows of $X$ and columns of $A$, respectively.  The
following canonical MMV formulation is very useful for our analysis.
\begin{definition}[Canonical form MMV \cite{Kim2010CMUSIC}]
Let $m$, $n$ and $r$  be a positive integers ($m<n$) that represents
the number of sensor elements, the ambient space dimension, and the
number of snapshots, respectively. Suppose that we are given a
multiple-measurement vector $B\in\mathbb{R}^{m\times r}$,
$X=[\mathbf{x}_1,\cdots,\mathbf{x}_r]\in\mathbb{R}^{n\times r}$, and
a sensing matrix $A\in \mathbb{R}^{m\times n}$.
A canonical form MMV problem
is given by the following optimization problem:
\begin{eqnarray}\label{eqdefcan_mmv}
{\rm minimize}~~~\|X\|_0\\
{\rm subject~to}~~~B=AX \notag,
\end{eqnarray}
where $\|X\|_0=|{\rm supp}X|$, ${\rm supp}X=\{1\leq i\leq n :
\mathbf{x}^i\neq 0\}$, $\mathbf{x}^i$ is the $i$-th row of $X$, and
the measurement matrix $B$ is full rank, i.e. ${\rm rank}(B)=r\leq
\|X\|_0$.
\end{definition}

Note that the canonical form MMV has the additional
constraint that ${\rm rank}(B)=r \leq \|X\|_0$. This is not
problematic at all since every MMV problem  can be converted into
canonical form using the singular value decomposition \cite{Kim2010CMUSIC}.
Now, the following theorem provides the $l_0$ sparse recovery bound from noiseless measurements.
\begin{theorem}[$l_0$ Bound]\cite{Feng97,chen2006trs,DaviesEldar2010}
Let ${\rm spark}(A)$ denote the smallest number
of linearly dependent columns of $A$. 
Then, $X\in\mathbb{R}^{n\times r}$  is the unique solution of $AX=B$ if and only if
\begin{equation}\label{l0-bound-mmv}
\|X\|_0<  \frac{{\rm spark}(A)+{\rm rank}(B)-1}{2} 
\leq {\rm
spark}(A)-1  \ .
\end{equation}
\end{theorem}

\section{Compressive MUSIC}\label{sec:review}

%
%
Consider a canonical form MMV problem. 
Suppose, furthermore, that
the columns of a sensing matrix $A\in\mathbb{R}^{m\times n}$ are in
general position; that is, any collection of $m$ columns of $A$ are
linearly independent. Then, according to \cite{Feng97,schmidt1986multiple},  for any $j\in \{1,\cdots,n\}$, $j\in {\rm
supp}X$ if and only if
\begin{equation}\label{music-cond}
Q^{*}\mathbf{a}_j=0,
\end{equation}
where $Q\in\mathbb{R}^{m\times (m-r)}$  consists of orthonormal
columns such that $Q^{*}B=0$ so that $R(Q)^{\perp}=R(B)$, which is
often called ``noise subspace''.
%

Note that the MUSIC criterion \eqref{music-cond} holds for all $m\geq
k+1$ if the columns of $A$ are in general position. Using the
compressive sensing terminology, this implies that the recoverable
sparsity level by MUSIC (with a probability 1 for the noiseless
measurement case) is given by
\begin{equation}\label{eq:max_music}
    \|X\|_0 < m = {\rm spark}(A)-1,
\end{equation}
where the last equality comes from the definition of the ${\rm spark}$.
Therefore, the $l_0$ bound \eqref{l0-bound-mmv} can be achieved by
MUSIC when $r=k$. However, for any $r<k$, the MUSIC condition
\eqref{music-cond} does not hold. This  is a major drawback of MUSIC
compared to the compressive sensing algorithms that allows perfect
reconstruction with extremely large probability by increasing the
sensor elements $m$. This drawback of the conventional MUSIC can be overcome by the following generalized MUSIC criterion \cite{Kim2010CMUSIC}.

\begin{theorem}\label{com-music} \cite{Kim2010CMUSIC}
Assume that $A\in\mathbb{R}^{m\times n}$, $X\in\mathbb{R}^{n\times
r}$, and $B\in\mathbb{R}^{m\times r}$ satisfy $AX=B$. Furthermore, we assume that
$\|X\|_0=k$ and $A$ satisfies the RIP condition with the left RIP constant $\delta^L_{2k-r+1}<1$.
If we are given $I_{k-r}\subset {\rm supp}X$
with $|I_{k-r}|=k-r$ and $A_{I_{k-r}}\in \mathbb{R}^{m\times(k-r)}$,
which consists of columns whose indices are in $I_{k-r}$, then for
any $j\in \{1,\cdots,n\}\setminus I_{k-r}$,
\begin{equation}\label{eq-comusicmod}
\ab_j^{*}\left[P_{R(Q)}-P_{R(P_{R(Q)}A_{I_{k-r}})}\right]\ab_j=0
\end{equation}
if and only if $j\in {\rm supp}X$.
\end{theorem}

Note that when $r=k$, the
condition (\ref{eq-comusicmod}) is the same as the MUSIC criterion
(\ref{music-cond}). By Theorem \ref{com-music}, we can develop
the Compressive MUSIC algorithm, which can be executed by these
processes.
\begin{itemize}
   \item Step 1: Find $k-r$ indices of ${\rm supp}X$ by any MMV compressive sensing algorithms such as 2-thresholding or SOMP.
   \item Step 2: Let $I_{k-r}$ be the set of indices which are taken in Step 1 and $S=I_{k-r}$.
   \item Step 3: For $j\in \{1,\cdots,n\}\setminus I_{k-r}$,  calculate the quantities $\eta(j)=\ab_j^{*}[P_{R(Q)}-P_{R(P_{R(Q)}A_{I_{k-r}})}]\ab_j$ for all $j\notin
   I_{k-r}$.
   \item Step 4:  Make an ascending ordering of $\eta(j)$, $j\notin
   I_{k-r}$ and choose indices that correspond to the first $r$
   elements and put these indices into $S$.
\end{itemize}

In compressive MUSIC, we determine $k-r$ indices of ${\rm supp}X$
with CS-based algorithms such as 2-thresholding or S-OMP, where the
exact identification of $k-r$ indices is a probabilistic matter. After that process,
we recover remaining $r$ indices of ${\rm supp}X$ with a generalized
MUSIC criterion, which is given in Theorem \ref{com-music}, and this reconstruction process is
deterministic. This hybridization makes the compressive MUSIC
applicable for all ranges of $r$, outperforming all the existing
methods.

So far, we introduced the compressive MUSIC algorithm. To analyze the performance of the compressive MUSIC, we find the number of measurements with which we can identify the support of $X$ by using compressive MUSIC with S-OMP. For this purpose, we consider the large system limit so that we assume the following conditions.

\begin{itemize}
\item Let $\rho:=\lim_{n\rightarrow\infty}m(n)/n$ exist. Then we call $\rho$ as the asymptotic under-sampling rate.
\item Let $\epsilon:=\lim_{n\rightarrow\infty}k(n)/n$ exist. Then we call $\epsilon$ as the asymptotic sparsity.
\item Let $\alpha:=\lim_{n\rightarrow\infty}r(n)/k(n)$ exist and $\gamma:=\lim_{n\rightarrow\infty}\sqrt{k/m}.$
\end{itemize}

Now, we may consider two cases according to the number of multiple measurement vectors.
First, we consider the case when the number of multiple measurement vectors are finite fixed number. Conventional compressive sensing (SMV problem) is a kind of this case. Second, we consider the case when $r$ is proportional to $n$. This case includes the conventional MUSIC case.  To analyze S-OMP, we assume that each element of $A$ is i.i.d. Gaussian random variable $\mathcal{N}(0,1/m)$.
In analyzing S-OMP, rather than analyzing the distribution of  $\|\ab_j^{*}P_{R(A_{I_t})}^{\perp}B\|_F^2$ where $I_t$ denotes the set of indices which are chosen in the first $t$ step of S-OMP,  we consider the following version of subspace S-OMP  due to its better performance \cite{DaviesEldar2010,LeeBresler2010}.

\begin{itemize}
\item Step 1 : Initialize $t=0$ and $I_0=\emptyset$.
\item Step 2 : Compute $P^\perp_{R(A_{I_t})}$ which is the projection operator onto the orthogonal complement of the span of $\{\ab_j:j\in I_t\}$.
\item Step 3 : Compute $P^\perp_{R(A_{I_t})}B$ and for all $j=1,\cdots, n$, compute
$\rho(t,j)=\|\ab_j^{*}P_{{R(P^\perp_{R(A_{I_t})}B)}}\|_F^2$.
\item Step 4 : Take $j_t=\arg\max_{j=1,\cdots,n}\rho(t,j)$ and $I_{t+1}=I_t\cup
\{j_t\}$ and if $t<k-r$ return to Step 2.
\item Step 5 : The final estimate of the $k-r$ elements of support is $I_{k-r}$.

\end{itemize}

\begin{theorem}\label{num-mea-somp}
Assume that we have multiple measurements $B=AX$ where each element of $A$ is generated from i.i.d. $\mathcal{N}(0,1/m)$ and $N$ is an additive noise. Then, in the large system limit, with probability 1, we can identify $k-r$ elements of the support of $X$ with subspace S-OMP if we have one of the following conditions :
\begin{itemize}
\item [1.] $r$ is a fixed finite number and
$$m>k(1+\delta)\frac{2\log{(n-k)}}{r}$$
for some $\delta>0$.
\item [2.] $r$ satisfies $\lim_{n\rightarrow\infty}(\log{n})/r=0$, $\lim_{n\rightarrow\infty}r/k=\alpha$ and
$$m>k(1+\delta)^2\left[2-F(\alpha)\right]^2$$
for some $\delta>0$ where
$$F(\alpha)=\frac{1}{\alpha}\int_0^{4t_1(\alpha)^2}xd\lambda_1(x),$$
$d\lambda_1(x)=(\sqrt{(4-x)x})/(2\pi x)$
is the probability measure with support $[0,4]$, $0\leq t_1(\alpha)\leq 1$ satisfies
$\int_0^{4t_1(\alpha)^2}d\lambda_1(x)=\alpha$.
\end{itemize} 
Here, $F(\alpha)$ is an increasing function on $
(0,1]$ such that $F(1)=1$ and $\lim_{\alpha\rightarrow0^{+}}F(\alpha)=0$.
\end{theorem}
\begin{proof}
See Appendix A. 
\end{proof}

By above theorem, the number of measurements for S-OMP shows some different characteristics according to the number of the measurement vectors. First, if we have small number of multiple measurement vectors, then the number of samples for S-OMP is reciprocally proportional to the number of multiple measurement vectors. On the other hand, we have sufficiently large number of snapshots such that $\lim_{n\rightarrow\infty}(\log{n})/r$ is close to 0, then the number of measurements for S-OMP varies from $4k$ to $k$ according to the ratio of $r$ and $k$ so that the $\log{n}$ is not necessary. In particular,  if the number of snapshots approaches the sparsity $k$, then we can identify the indices of ${\rm supp}X$ with only $(1+\delta)k$ where $\delta$ is any small positive number, which is equivalent to the required number of multiple measurement vectors for the success of conventional MUSIC algorithm. 

Furthermore, in \cite{Kim2010CMUSIC}, we developed the analysis for the noisy setting, where we showed that the required SNR for the success of support recovery decreases when the asymptotic ratio of the number of snapshots and the sparsity level (that is, $\lim_{n\rightarrow\infty}r/k$) increases, in the large system limit. This is one of the important advantages of MMV over SMV.

\section{Optimized partial support selection}\label{sec:no}

As discussed before,  we can easily expect that the performance of the compressive MUSIC is very dependent on the selection of $k-r$ correct indices of the support of $X$.  Note that this is a very stringent condition.  In practice, even though the consecutive $k-r$ steps of S-OMP may not be correct, there are chances that among the $k$-sparse solution of S-OMP, part of the supports can be correct.
 Hence, if the estimate of the support of $X$ has at least $k-r$ indices of the support of $X$ and we can identify them, then we can expect that the performance of the compressive MUSIC will be improved. When ${k \choose k-r}$ is small, we may apply the exhaustive search, but if both $k-r$ and $r$ are not small, then the exhaustive search is hard to apply so that we have to find some alternative method to identify the correct indices from the estimate of ${\rm supp}X$.
Indeed, the following subspace fitting criterion can address the problem.
\begin{theorem}\label{thm-sfcriterion}
Assume that we have a canonical MMV model $AX=B$ where $A\in\mathbb{R}^{m\times n}$, $X\in\mathbb{R}^{n\times r}$, $\|X\|_0=k$ and $r<k<m<n$. If there is an index set $I_k\subset \{1,\cdots,n\}$ such that $|I_k|=\min\{k, {\rm spark}(A)-r\}$ and $|I_k\cap {\rm supp}X|\geq k-r+1$, then for any $j\in I_k$,
$j\in {\rm supp}X$ if and only if
\begin{equation}\label{sf-criterion}
P_{Q_{k,j}}\ab_j={\bf 0},
\end{equation}
where $Q_{k,j}$ is the orthogonal complement for $R([B~~A_{I_k\setminus \{j\}}])$, $A_{I_k\setminus \{j\}}$ consists of columns of $A$ whose index belongs to $I_k\setminus \{j\}$ and $P_{R([B~~A_{I_k\setminus \{j\}}])}^{\perp}$ is the orthogonal projection on $R([B~~A_{I_k\setminus \{j\}}])^{\perp}$.
\end{theorem}
\begin{proof}
Assume that $j\in I_k\cap {\rm supp}X$. Then $|(I_k\setminus\{j\})\cap{\rm supp}X|\geq k-r$ so that
$$R([B~~A_{I_k\setminus\{j\}}])\supseteq R([B~~A_{J_{k-r}}])\cap
R(A_S)=R(A_S)$$
where $J_{j,k-r}\subset (I_k\setminus \{j\})\cap S$, $|J_{j,k-r}|=k-r$ and $S={\rm supp}X$.
Since $\ab_j\in R(A_S)$, (\ref{sf-criterion}) holds for $j\in I_k\cap {\rm supp}X$.

To show the converse, assume that (\ref{sf-criterion}) holds for some $j\in I_k$. Then we have $\ab_j\in R([B~~A_{I_k}\setminus \{j\}])$, that is, there some $\pb\in \mathbb{R}^r$ and $\qb\in\mathbb{R}^{|I_k|-1}$ such that
$$\ab_j=B\pb+A_{I_k\setminus \{j\}}\qb=AX\pb+A_{I_k\setminus\{j\}}\qb.$$
Since $|({\rm supp}X)\cup I_k|\leq k+|I_k|-(k-r+1)\leq k+{\rm spark}(A)-r-(k-r+1)={\rm spark}(A)-1,$
if $j\notin {\rm supp}X$, then there is an $\rb\in \mathbb{R}^n\setminus \{0\}$ such that
$\|\rb\|_0<{\rm spark}(A)$ and $A\rb=0$ since $j\notin {\rm supp}X\cup (I_k\setminus \{j\})$. Then, by the definition of spark$(A)$, that is a contradiction so that $j\in {\rm supp}X$ if (\ref{sf-criterion}) holds.
\end{proof}

In particular, if the columns of $A$ are in general position, then we can take index set $I_k$ with $|I_k|=\min\{k,m-r+1\}$. Also, if $A$ has an RIP condition with $\delta_{2k}<1$, then we can take $|I_k|=k$ since $r\leq k$.
Then, Theorem~\ref{thm-sfcriterion} informs us that  we only require the partial support recovery rather than $k-r$ consecutive correct CS step \cite{Kim2010CMUSIC}.
Accordingly, the compressive MUSIC with optimized partial  support is then performed by following procedure.

\begin{itemize}
\item step 1 : Let $S=\emptyset$.
\begin{itemize}
\item If $r<k$, estimate $k$ indices of ${\rm supp}X$ by MMV compressive sensing algorithm.
\item If $r=k$, goto step 5.
\end{itemize}
\item step 2 : Let $I_{k}$ be the set of indices which are taken in step 1.
\item step 3 : For $j\in I_{k}$, calculate the quantities $\zeta(j)=
\|P_{Q_{k,j}}\ab_j\|^2.$
\item step 4 : Make an ascending ordering of $\zeta(j)$, $j\in I_k$ and choose indices that corresponds the first $k-r$ elements and put these indices into $S$.
\item step 5 : For $j\in \{1,\cdots,n\}\setminus S$, calculate the quantities
$\eta(j)=\gb_j^{*}P_{G_{I_{k-r}}}^{\perp}\gb_j$.
\item step 6 : Make an asending ordering of $\eta(j)$, $j\notin S$ and choose indices that correspond to the first $r$ elements and put these indices into $S$.
\end{itemize}
In above algorithm, we require partial correctness of support estimation instead of exactness of $k-r$ consecutive support estimation. Moreover, the step 1 in above algorithm need not to be  greedy so that we can also apply the convex optimization algorithm such as $l_{2,1}$ minimization \cite{malioutov2005ssr} or belief propagation \cite{KCJBY2011}.

So far, we have assumed that the measurement $B$ is without noise. For the case of noisy measurement, $B$ is corrupted so that the optimized partial support  selection is affected by noise. Although we do not discuss the noise sensitivity in this paper, this issue will be investigated in the future works.

\section{Information theoretic analysis for partial support recovery for MMV}
\label{}
From above section, we know that compressive MUSIC with optimized partial support  can bear with  the fractional  distortion  of  support estimate error less than $\alpha$ to guarantee the exact recovery   in the large system limit. Therefore,   in this section, we are interested in finding a sufficient condition such that we can find the estimate for the support with fractional distortion less than $\alpha$ in an MMV step.
Here, we consider the linear model in which the multiple measurement $Y\in\mathbb{R}^{m\times r}$ is given as
$$Y=AX+N$$
where $A\in\mathbb{R}^{m\times n}$ is a sensing matrix and $N\in\mathbb{R}^{m\times r}$ is additive noise whose columns are i.i.d. and have the distribution $\mathcal{N}(0,\sigma_w^2I)$. Also we assume that $X$ has $k$ nonzero rows which are indexed by the set $S$ and that $S$ is distributed uniformly over the ${n \choose k}$ possibilities. Again, we assume that the distributions of each column of $X$ are identical and independent. Furthermore, we assume that the elements of sensing matrix $A$ are randomly given with i.i.d. $\mathcal{N}(0,1/n)$. Here we consider the large system limit. Also, we use the following definition for SNR.

\begin{definition}
For a given multiple signal $X$, the SNR is given by
$${\sf SNR}(X)=\frac{{\sf E}[\|AX\|_F^2]}{{\sf E}[\|N\|_F^2]}=
\frac{\|X\|_F^2}{rn\sigma_w^2}.$$ Also, for a stochastic signal class $\mathcal{X}$, ${\sf SNR}(\mathcal{X})$ is called an asymptotic lower bound on ${\sf SNR}(X)$ if there exists a constant $c>0$ such that
$${\sf P}\{{\sf SNR}(\mathcal{X}(n))\leq {\sf SNR}(X)\}>1-e^{-nc}.$$
\end{definition}

The analysis for partial support recovery use an information theoretic approach which was used in \cite{ReevesGastpar2008} so that we define the following function.
\begin{definition}
    For $p\in [0,1]$ and $u\in [0,1-\epsilon]$, we define
    $$h(\epsilon,\alpha)=\epsilon h(\alpha)+(1-\epsilon)h\left(\frac{\alpha}{1/\epsilon-1}\right),$$
where $h(p)=-p\log{p}-(1-p)\log{(1-p)}$ is the binary entropy function.
\end{definition}

For a fractional distortion $\alpha>0$, we define the fractional partial recovery with distortion rate $\alpha$ by the requirement $d(S,\hat{S})/k\leq \alpha$ where $\hat{S}$ is the estimate for the support of $X$ such that $|\hat{S}|=k$ and $d(S,\hat{S})=|S\setminus \hat{S}|$. If $\alpha>1-\epsilon$, the random guessing estimator $\hat{K}_{RG}$ is asymptotically reliable so that we assume that $\alpha\leq 1-\epsilon$ \cite{ReevesGastpar2008}.

For the analysis, we consider the maximum likelihood (ML) estimator which is given by
$$\hat{S}_{ML}(Y)=\arg\min_{|U|=k}\|P_{R(A_U)}^{\perp}Y\|_F^2$$
where $P_{R(A_U)}^{\perp}$ is the projection operator onto the orthogonal complement of $R(A_U)$. For $k$-sparse multiple input signal $X\in \mathbb{R}^{n\times r}$, we introduce the following term.
\begin{definition}
    Let $\Zb$ correspond to the nonzero rows of $X$ and satisfy $\|\zb^1\|_2\leq \|\zb^2\|_2\leq\cdots\leq\|\zb^k\|_2$. Then, for some $\alpha\leq 1$, we let
    $$g(\alpha,X)=\frac{1}{\alpha\|X\|_F^2}\sum\limits_{i=1}^{[\alpha k]}\|\zb^i\|_2^2.$$
   Also, for a stochastic signal class $\mathcal{X}$, let $g(\alpha,\mathcal{X})$ be the asymptotic lower bound on $g(\alpha,X)$ if there is a constant $c>0$ such that
   $${\sf P}\{g(\alpha,\mathcal{X}(n))\leq g(\alpha,X)\}>1-e^{-nc}.$$
\end{definition}

In \cite{ReevesGastpar2008},  Reeves and Gastpar gave sufficient conditions for partial support recovery for SMV problem using ML estimator. We can extend those results to the MMV problem as the following theorem.
\begin{theorem}\label{ML-suffcond}
For a given signal class $\mathcal{X}$, sparsity $\epsilon\in (0,1)$, undersampling ratio $\rho<1$, the fractional distortion $\alpha\in (0,1-\epsilon)$, the estimator $\hat{S}_{ML}$ is asymptotically reliable if
\begin{equation}\label{par-snr}
{\sf SNR}(\mathcal{X})>\frac{1}{\alpha g(\alpha,\mathcal{X})}
\end{equation}
and
\begin{equation}\label{par-sam}
\rho>\epsilon + \frac{1}{r}\max\limits_{u\in [\alpha,1-\epsilon]}
\frac{2h(\epsilon,u)}{\log{(\gamma(u,\mathcal{X}))}+\gamma(u,\mathcal{X})^{-1}-1}
\end{equation}
where $\gamma(u,\mathcal{X})={\sf SNR}(\mathcal{X})ug(u,\mathcal{X})$.
\end{theorem}
\begin{proof}
See Appendix~B.
\end{proof}
\bigskip

Note that if $\alpha>1-\epsilon$, the random guessing estimator is asymptotically reliable so that we can identify the support with distortion less than $\alpha$ with large probability in the large system case, by augmenting randomly chosen $k-r$ support in generalized MUSIC step. Moreover, in this case, the sufficient condition becomes $\rho>\epsilon$, which is equivalent to the MUSIC for the full rank measurement.

In addition, in \cite{ReevesGastpar2008}, Reeves and Gastpar gave necessary conditions for partial support recovery for SMV problem. The counterpart for MMV can be given      
by the following theorem. For the proof, see Appendix C. 

\begin{theorem}\label{nec-cond}
For a given stochastic signal class $\mathcal{X}$, sparsity $\epsilon\in (0,1)$, sampling rate $\rho<1$ and fractional distortion $\alpha\in (0,1-\alpha)$, a necessary condition for asymptotically reliable recovery is 
$$\rho>\frac{h(\epsilon)-h(\epsilon,\alpha)+I(X;Y|S)/n}{\sum\limits_{l=1}^r
\frac{1}{2}\log{(1+\frac{1}{\sigma_w^2}\kappa_l(\mathcal{X}))}},$$
where $I(X;Y|K)$ is the mutual information between $X$ and $Y$ conditioned on $S$, and $\kappa_l(\mathcal{X})$ is the asymptotic upper bound for the $l$-th largest eigenvalue of $X^{*}X$, where $X\in \mathcal{X}$.
\end{theorem}

\section{Numerical Simulation}
\label{sec:simulation}
We compared the performance of compressive MUSIC with optimized partial support (proposed algorithm), compressive MUSIC (CS-MUSIC), subspace-augmented MUSIC (SA-MUSIC) and S-OMP.  We used S-OMP as a MMV compressive sensing algorithm for various hybrid MMV algorithms.  In order to quantify the performance of each algorithms, the empirical recovery ratio is calculated which is defined as the percentage of correct identification of all supports, and the ratio are averaged for $5000$ simulation results. The simulation parameters are as following: $m=40$, $n=100$, the number of measurement vectors  is $r=9$, and $k=1,2,\cdots,20$, respectively. Each component of the sensing matrix $A$ is generated by i.i.d. Gaussian random variable $\frac{1}{\sqrt{m}}\mathcal{N}(0,1)$ or $\frac{1}{\sqrt{m}}\mathcal{N}(1,1)$ to see the effect of RIP in each algorithms. Gaussian noise of ${\sf SNR}=40dB$ is added to the measurement vector $B$.
In Figure \ref{fig:bound_CMUSIC_SOMP1}, we can observe that the proposed method shows significantly better performance than the original version of compressive MUSIC, SA-MUSIC and S-OMP.  In particular,  the proposed method is more robust to bad RIP of the sensing matrices  such that the performance gain is more prominent.
\begin{figure}[htbp]
 \centerline{
 \epsfig{figure=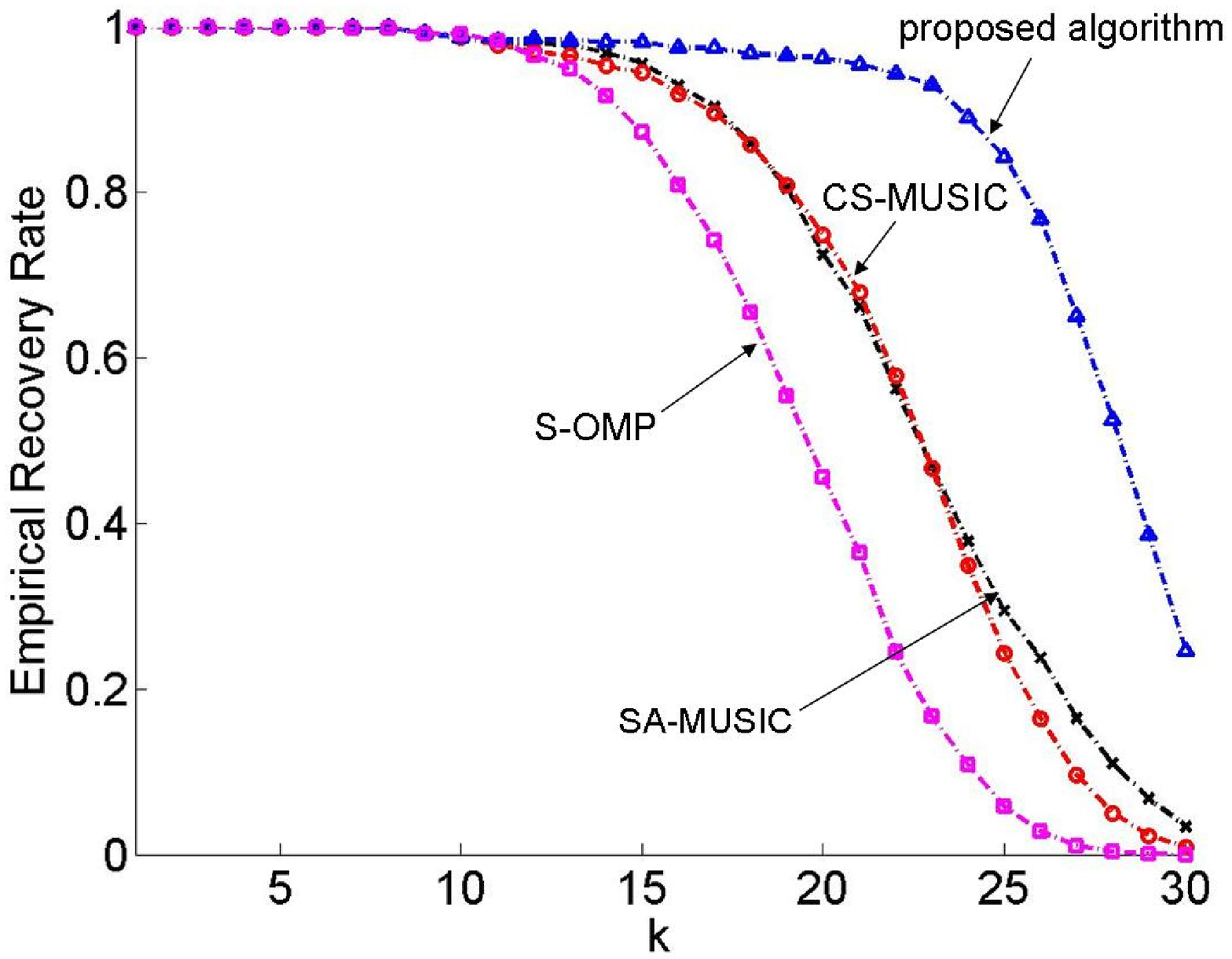,width=5cm}}
 \centerline{ \mbox{(a)}}
 \centerline{\epsfig{figure=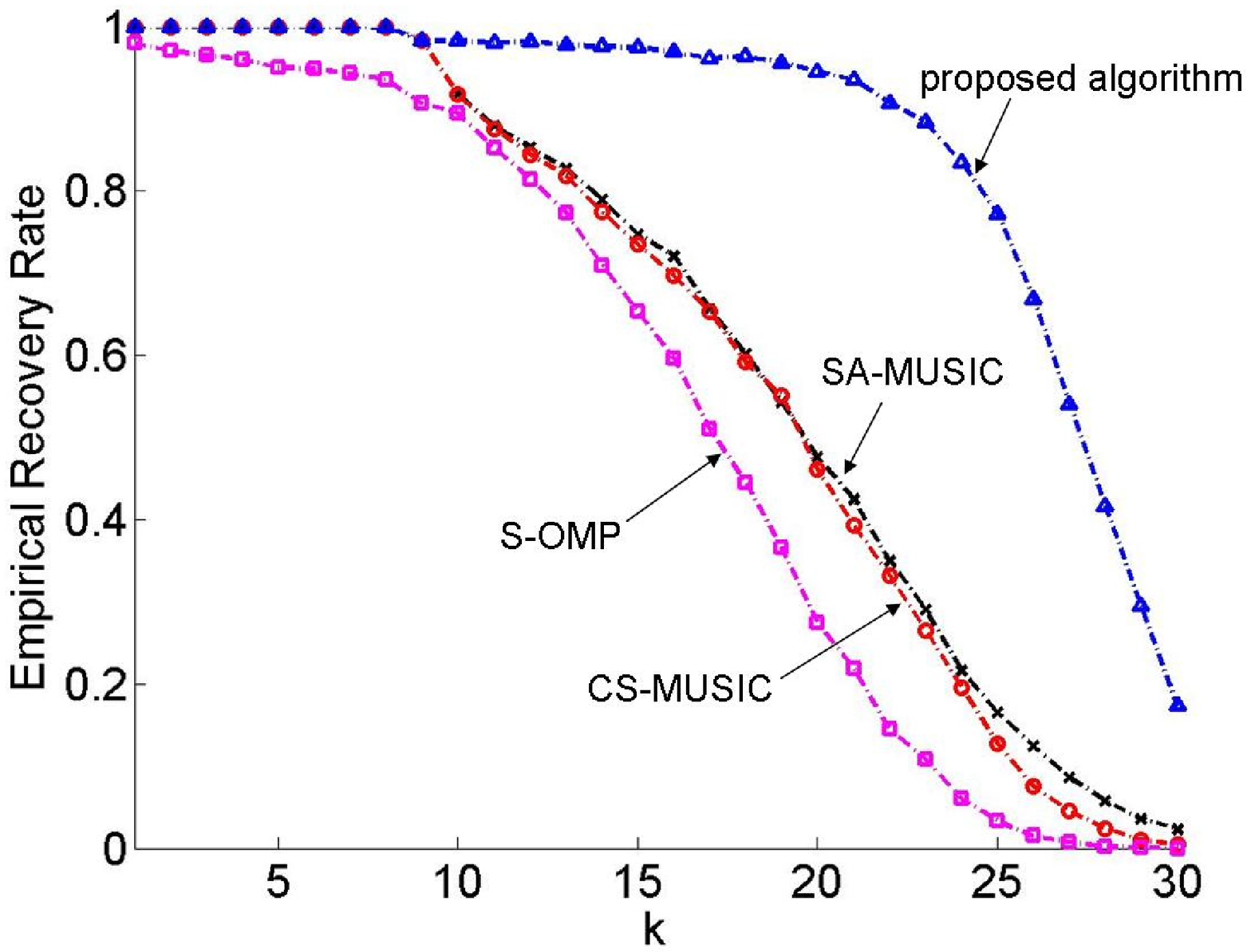,width=5cm}}
  \centerline{ \mbox{(b)}}
  \caption{Recovery rates  when $m=40, r=9$, ${\sf SNR}=40dB$ and $A$ is generated from (a) $\frac{1}{\sqrt{m}}\mathcal{N}(0,1)$ and (b) $\frac{1}{\sqrt{m}}\mathcal{N}(1,1)$.}
 \label{fig:bound_CMUSIC_SOMP1}
\end{figure}

\section{Conclusion}
This paper proposed a mathematical framework for optimized partial support selection to improve the performance of compressive MUSIC for joint sparse recovery. We first discussed about the original compressive MUSIC algorithm, and derived the sharp bound for the number of measurement for exact recovery using subspace S-OMP for partial support recovery. Then, we discussed that the requirement of the correct $k-r$ step S-OMP can be relaxed such that as long as $k-r+1$ supports from $k$-support estimate are correct,  subspace fitting criterion can identify the correct $k-r$ support to improve the robustness of the compressive MUSIC algorithm.  Information theoretical analysis was also provided to obtain a sufficient condition for MMV joint sparse recovery using compressive MUSIC algorithm. As a future work, we will derive the SNR condition for the success of subspace fitting step.


\section*{Acknowledgment}
 This research was  supported by the Korea Science and Engineering Foundation
(KOSEF) grant funded by the Korean government (MEST) (No.
2009-0081089).

\bibliographystyle{plain}
\bibliography{cites,bispl2010,totalbiblio_bispl}

\begin{thebibliography}{10}

\bibitem{chen2006trs}
J.~Chen and X.~Huo.
\newblock {Theoretical results on sparse representations of multiple
  measurement vectors}.
\newblock {\em IEEE Trans. on Signal Processing}, 54(12):4634--4643, 2006.

\bibitem{DaviesEldar2010}
M.E. Davies and Y.C. Eldar.
\newblock {Rank awareness for joint sparse recovery}.
\newblock {\em arXiv:1004.4529}, 2010.

\bibitem{Feng97}
P.~Feng.
\newblock {\em Universal minimum-rate sampling and spectrum-blind
  reconstruction for multiband signals}.
\newblock Dissertation, University of Illinois, Urbana-Champaign, 1997.

\bibitem{fletcher2009necessary}
A.K. Fletcher, S.~Rangan, and V.K. Goyal.
\newblock {Necessary and sufficient conditions for sparsity pattern recovery}.
\newblock {\em IEEE Trans. on Information Theory}, 55(12):5758--5772, 2009.

\bibitem{Kim2010CMUSIC}
J.M. Kim, O.K. Lee, and J.C. Ye.
\newblock Compressive {MUSIC}: a missing link between compressive sensing and
  array signal processing.
\newblock {\em arXiv:1004.4398}, 2010.

\bibitem{KLY2011opt-CSMUSIC}
Jong~Min Kim, Ok~Kyun Lee, and Jong~Chul Ye.
\newblock {Compressive MUSIC with optimized partial support for joint sparse
  recovery}.
\newblock {\em arXiv:1102.3288}, 2011.

\bibitem{KCJBY2011}
Jongmin Kim, Woohyuk Chang, Bangchul Jung, Dror Baron, and Jong~Chul Ye.
\newblock {Belief propagation for joint sparse recovery}.
\newblock {\em arXiv:1102.3289}, 2011.

\bibitem{LeeBresler2010}
K.~Lee and Y.~Bresler.
\newblock {Subspace-augmented MUSIC for joint sparse recovery}.
\newblock {\em arXiv:1004.4371}, 2010.

\bibitem{malioutov2005ssr}
D.~Malioutov, M.~Cetin, and A.S. Willsky.
\newblock {A sparse signal reconstruction perspective for source localization
  with sensor arrays}.
\newblock {\em IEEE Trans. on Signal Processing}, 53(8):3010--3022, 2005.

\bibitem{marcenko1967distribution}
V.A. Mar{\v{c}}enko and L.A. Pastur.
\newblock {Distribution of eigenvalues for some sets of random matrices}.
\newblock {\em Sbornik: Mathematics}, 1(4):457--483, 1967.

\bibitem{Reeves2008}
Galen Reeves.
\newblock {Sparse signal sampling using noisy linear projections}.
\newblock {\em Technical Report No. UCB/EECS-2008-3}, 2008.

\bibitem{ReevesGastpar2008}
Galen Reeves and Michael Gastpar.
\newblock Sampling bounds for sparse support recovery in the presence of noise.
\newblock In {\em IEEE International Symposium on Information Theory}, pages
  2187--2191, Toronto, Canada, July 6-11 2008.

\bibitem{schmidt1986multiple}
R.~Schmidt.
\newblock {Multiple emitter location and signal parameter estimation}.
\newblock {\em IEEE Trans. on Antennas and Propagation}, 34(3):276--280, 1986.

\end{thebibliography}
\newpage



%
\section*{Appendix A: Number of measurements for Compressive MUSIC with subspace S-OMP}
In this section, we assume the large system limit so that we will assume that $\rho$, $\epsilon$, $\alpha$ and $\gamma$ exist. In this section, we will use the following theorem, which gives us the asymptotic distribution of singular values for Gaussian random matrices.
\begin{theorem}\cite{marcenko1967distribution}
Suppose that each entry of $A\in\mathbb{R}^{m\times k}$ is generated from i.i.d. Gaussian random variable $\mathcal{N}(0,1/m)$. Then the probability density of squared singular value of $A$ is given by
\begin{equation}\label{marpas}
d\lambda_{\gamma}(x):=\frac{1}{2\pi\gamma^2}\frac{\sqrt{((1+\gamma)^2-x)(x-(1-\gamma)^2)}}{x}dx.
\end{equation}
\end{theorem}

By using above theorem, we prove the Theorem \ref{num-mea-somp}.

{\em Proof of Theorem \ref{num-mea-somp}:} First, for $j\notin {\rm supp}X$, since $\ab_j$ is statistically independent of $P_{R(A_{I_t})}^{\perp}B$. For $t\leq k-r$, the dimension of $P_{R(A_{I_t})}^{\perp}B$ is $r$ so that $m\|\ab_j P_{R(P_{R(A_{I_t})}^{\perp}B)}\|^2$ is of chi-squared distribution of degree of freedom $r$.

On the other hand, for $j\in {\rm supp}X$, we have
\begin{eqnarray*}
\max\limits_{j\in{\rm supp}X}\|\ab_j^{*}P_{R(P_{R(A_{I_t})}^{\perp}B)}\|^2&\geq&
\frac{1}{k}\|A_S^{*}P_{R(P_{R(A_{I_t})}^{\perp}B)}A_S\|_F^2\\
&\geq&\frac{\sum\limits_{j=1}^r \sigma_j^2(A_S)}{k}
\end{eqnarray*}
since $R(P_{R(A_{I_t})}^{\perp}B)\subset R(A_S)$, where $A_S$ have singular values
$0<\sigma_1\leq \sigma_2\leq\cdots\leq \sigma_k$. Then by (\ref{marpas}), we have
\begin{eqnarray}\label{mpd1}
\lim\limits_{n\rightarrow\infty}\frac{\sum\limits_{j=1}^r \sigma_j^2(A_S)}{k}=
\alpha\frac{\int_{(1-\gamma)^2}^{(1-\gamma+2\gamma t_\gamma(\alpha))^2}xd\lambda_{\gamma}(x)}{\alpha}
\end{eqnarray}
where $0\leq t_{\gamma}(\alpha)\leq 1$ is the value satisfying
$$\int_{1-\gamma}^{1-\gamma+2\gamma t_{\gamma}(\alpha)}ds_{\gamma}(x)
=\int_{(1-\gamma)^2}^{(1-\gamma+2\gamma t_{\gamma}(\alpha))^2}d\lambda_{\gamma}(x)=\alpha.$$
If we let 
\begin{eqnarray}
d\lambda_{0,\gamma}(x)=\frac{1}{\pi}\frac{\sqrt{4-s}\sqrt{s}}{\gamma s+(1-\gamma)^2}dx,
\end{eqnarray}
then we have for any $0\leq t\leq 4$,
\begin{eqnarray}
\int_0^t d\lambda_1(x)\leq \int_0^t d\lambda_{0,\gamma}(x)
\end{eqnarray}
then by substitution with $s=(x-(1-\gamma)^2)/\gamma$, we have
\begin{eqnarray}\label{mpd2}
\notag &&\int_{(1-\gamma)^2}^{(1-\gamma+2\gamma t_{\gamma}(\alpha))^2}xd\lambda_{\gamma}(x)\\
\notag &=&\int_0^{\frac{(1-\gamma+2\gamma t_{\gamma}(\alpha))^2-(1-\gamma)^2}{r}}[(1-\gamma)^2+\gamma s]d\lambda_{0,\gamma}(s)\\
\notag &\geq&\int_0^{4t_1(\alpha)^2}[(1-\gamma)^2+\gamma s]d\lambda_1(s)\\
&=&(1-\gamma)^2\alpha+\gamma\int_0^{4t_1(\alpha)^2}sd\lambda_1(s),
\end{eqnarray} 
where the inequality comes from Lemma \ref{lem-t}.
Substituting (\ref{mpd2}) into (\ref {mpd1}), we have
\begin{eqnarray*}
&&\lim\limits_{n\rightarrow\infty}\frac{\sum\limits_{j=1}^r \sigma_j^2(A_S)}{k}\\
&\geq&\alpha\left[(1-\gamma)^2+\gamma \frac{\int_0^{4t_1(\alpha)^2}sd\lambda_1(s)}{\alpha}\right]\\
&=&\frac{r}{m}(1/\gamma-1)^2+\alpha\gamma F(\alpha)
\end{eqnarray*}
where $$F(\alpha):=(1/\alpha)\int_0^{4t_1(\alpha)^2}sd\lambda_1(s)=\frac{\int_0^{4t_1(\alpha)^2}sd\lambda_1(s)}{\int_0^{4t_1(\alpha)^2}d\lambda_1(s)}$$ is an increasing function with respect to $\alpha$ such that $\lim_{\alpha\rightarrow 0}F(\alpha)=0$ and $\alpha(1)=1$.\\
Then we consider two limiting cases according to the number of measurement vectors.\\
{\bf  (Case 1)} For $t\leq k-r$, $\{m\|\ab_j^{*}P_{R(P_{R(A_{I_t})}^{\perp}B)}\|^2:j\notin {\rm supp}X\}$ are independent chi-squared random variables of degree of freedom $r$ so that by Lemma \ref{chi-max}, we have
\begin{equation}\label{gamma1}
\lim_{n\rightarrow\infty}\max\limits_{j\notin {\rm supp}X}\frac{m\|\ab_j^{*}P_
{R(P_{R(A_{I_t})}^{\perp}B)}\|^2}{2\log{(n-k)}}=1.
\end{equation}
Here we assume that
\begin{equation}\label{num-somp-rfixed}
m>k\frac{2(1+\delta)\log{(n-k)}}{r}.
\end{equation}
Then by Mar\'{c}enko-Pastur theorem \cite{marcenko1967distribution},
\begin{eqnarray*}
\lim\limits_{n\rightarrow\infty}\sigma_{\min}(A_S)&=&\lim\limits_{n\rightarrow\infty}
(1-\sqrt{k/m})^2\\&\geq& \lim\limits_{n\rightarrow\infty}\left(1-\sqrt{r/(2\log{(n-k)})}\right)^2=1
\end{eqnarray*}
so that
\begin{eqnarray}\label{j-supp-lower2}
\notag &&\liminf\limits_{n\rightarrow\infty}\max\limits_{j\in{\rm supp}X}\frac{m\|\ab_j^{*}P_{R(P_{R(A_{I_t})}^{\perp}B)}\|^2}{2\log{(n-k)}}\\
\notag&\geq&\liminf\limits_{n\rightarrow\infty}\frac{m}{2\log{(n-k)}}\frac{\sum\limits_{j=1}^r \sigma_{k-j+1}^2(A_S)}{k}\\&\geq& \liminf\limits_{n\rightarrow\infty}\frac{r}{2\log{(n-k)}}
\left(\frac{1}{\gamma}\right)^2\geq 1+\delta.
\end{eqnarray}
Hence, when $r$ is a fixed number, if we have \eqref{num-somp-rfixed}, then we can identify $k-r$ correct indices of ${\rm supp}X$ with subspace S-OMP, in the large system limit.
\\
\indent {\bf (Case 2)} Similarly as in the previous case, for $t< k-r$, $\{m\|\ab_j^{*}P_
{R(P_{R(A_{I_t})}^{\perp}B)}\|^2:j\notin {\rm supp}X\}$ are independent chi-squared  distribution.  Since $\lim_{n\rightarrow\infty}(\log{n})/r=0$, by Lemma 3 in \cite{fletcher2009necessary}, we have
\begin{equation}\label{gamma2}
\lim_{n\rightarrow\infty}\max\limits_{j\notin {\rm supp}X}\frac{m\|\ab_j^{*}P_
{R(P_{R(A_{I_t})}^{\perp}B)}\|^2}{r}=1.
\end{equation}
On the other hand, for $j\in {\rm supp}X$, we have
\begin{eqnarray}\label{omp-supp2}
\notag &&\liminf_{n\rightarrow\infty}\max\limits_{j\in {\rm supp}X}\frac{m\|\ab_j^{*}P_
{R(P_{R(A_{I_t})}^{\perp}Y)}\|^2}{r}\\
&\geq&\left(\frac{1}{\gamma}-1\right)^2+F(\alpha).
\end{eqnarray}
We let
\begin{equation}\label{num-somp2}
m>k(1+\delta)^2\left[2-F(\alpha)\right]^2
\end{equation}
for some $\delta>0$. Note that \eqref{num-somp2} is equivalent to
\begin{equation*}
\frac{1}{\gamma}>(1+\delta)[2-F(\alpha)]
\end{equation*}
Again we let
$$u:=F(\alpha)~{\rm and}~v:=\frac{4}{\alpha}\frac{\kappa(B)+1}{{\sf SNR}_{\min}(B)-1}.$$
Then for a quadratic function $Q(x)=(x-1)^2+ux$, if $x>(1+\delta)(2-u)$, then we have
\begin{eqnarray}\label{aux-ineq2}
\notag Q(x)&=&x^2-(2-u)x+1=x\left[x-(2-u)\right]+1\\
&>&\delta(1+\delta)(2-u)^2+1\geq  1+\delta(1+\delta)
\end{eqnarray}
since $0\leq u\leq 1$. Combining \eqref{omp-supp2} and \eqref{aux-ineq2}, we have for $0\leq t<k-r$ and $j\in {\rm supp}X$, we have
$$\liminf_{n\rightarrow\infty}\max\limits_{j\in {\rm supp}X}\frac{m\|\ab_j^{*}P_
{R(P_{R(A_{I_t})}^{\perp}Y)}\|^2}{r}\geq 1+\delta(1+\delta)$$
for some $\delta>0$. Hence, in the case of $\lim_{n\rightarrow\infty} r/k=\alpha>0$, we can identify the correct indices of ${\rm supp}X$ if we have (\ref{num-somp2}).

\begin{lemma}\label{lem-t}
For $0\leq \gamma\leq 1$ and $0\leq \alpha\leq 1$, we let $0\leq t_{\gamma}(\alpha)\leq 1$ which satisfies 
$$\int_{(1-\gamma)^2}^{(1-\gamma+2\gamma t_{\gamma}(\alpha))^2}ds_{\gamma}(x)=\alpha$$
where $d\lambda_{\gamma}(x)$ is the probability measure which is given by
$$d\lambda_{\gamma}(x):=\frac{1}{\pi \gamma^2}
\frac{\sqrt{((1+\gamma)^2-x)(x-(1-\gamma)^2)}}{x}.$$
Furthermore, we let $d\lambda_{0,\gamma}(x)$ is the probability measure which is given by 
$$d\lambda_{0,\gamma}(x)=\frac{1}{\pi}\frac{\sqrt{4-x}\sqrt{x}}{\gamma x+(1-\gamma)^2}dx.$$
Then we have 
\begin{eqnarray*}
&&\int_0^{\frac{(1-\gamma+2\gamma t_{\gamma}(\alpha))^2-(1-\gamma)^2}{\gamma}}[(1-\gamma)^2+\gamma x]d\lambda_{0,\gamma}(x)\\
&\geq&\int_0^{4t_1(\alpha)^2}[(1-\gamma)^2+\gamma x]d\lambda_1(x).
\end{eqnarray*}
\end{lemma}

For the proof of Lemma \ref{lem-t}, we need the following lemma. 

\begin{lemma}\label{lem-fg}
Let $-\infty<a<b<\infty$. Suppose that $f_1(x)$ and $f(x)$ are continouous probability density functions on $[a,b]$ such that for any $t\in [a,b]$, 
$$\int_a^t f(x)dx\geq \int_a^t f_1(x)dx,$$ and satisfy that
$$f_1(x)>0 ~{\rm and}~ f(x)>0 ~{\rm on}~(a,b).$$
Then for any nonnegative increasing function $g(x)$ on $[a,b]$ and for any $(q_1,q)\in [a,b]\times [a,b]$ such that 
\begin{equation}\label{samearea}
\int_a^{q_1}f_1(x)dx=\int_a^q f(x)dx,
\end{equation}
we have 
\begin{eqnarray}\label{fg}
\int_a^{q_1}g(x)f_1(x)dx\geq \int_a^{q}g(x)f(x)dx.
\end{eqnarray}
\end{lemma}
\begin{proof}
First, we define 
\begin{equation*}
F_1(x)=\int_a^x f_1(t)dt~{\rm and}~F(x)=\int_a^x f(t)dt.
\end{equation*}
Then both $F_1(x)$ and $F(x)$ are strictly increasing functions so that their inverse functions exist and satisfy 
$F_1^{-1}(x)\geq F^{-1}(x)$ for any $x\in [0,1]$. For any $(q_1,q)\in [a,b]\times[a,b]$ which satisfies \eqref{samearea}, there is some $c\in [0,1]$ such that 
$F_1(q_1)=F(q)=c.$
Applying the change of variable, we have 
\begin{eqnarray*}
&&\int_a^{q_1}g(x)f_1(x)dx-\int_a^q g(x)f(x)dx\\
&=&\int_0^c [g(F_1^{-1}(s))-g(F^{-1}(s))]ds\geq 0
\end{eqnarray*}
since $F_1^{-1}(x)\geq F^{-1}(x)$ for any $x\in [0,1]$ and $g(x)$ is increasing on $[a,b]$.
\end{proof}

{\em Proof of Lemma \ref{lem-t}}
Noting that we have 
\begin{eqnarray*}
\alpha&=&\int_{0}^{\frac{(1-\gamma+2\gamma t_{\gamma}(\alpha))^2-(1-\gamma)^2}{\gamma}}
d\lambda_{0,\gamma}(s)
=\int_0^{4t_1(\alpha)^2}d\lambda_1(s),
\end{eqnarray*}
by Lemma \ref{lem-fg}, we only need to show that 
\begin{equation*}
\int_0^t d\lambda_{0,\gamma}(s)\geq \int_0^t d\lambda_1(s) 
\end{equation*}
for any $t\in [0,4]$.
Let $f_1(x)$ and $f(x)$ be given by
\begin{eqnarray*}
f_1(x)&=&\frac{1}{\pi}\frac{\sqrt{4-x}\sqrt{x}}{\gamma x+(1-\gamma)^2},\\
f(x)&=&\frac{1}{\pi}\frac{\sqrt{4-x}\sqrt{x}}{x}.
\end{eqnarray*}
Then we can see that 
\begin{eqnarray*}
f(x)\geq f_1(x) &&{\rm for}~x\in (0,1-\gamma)\\
~{\rm and}~f_1(x)\geq f(x)&&{\rm for}~x\in [1-\gamma,1].
\end{eqnarray*}
Since $f_1(x)$ and $f(x)$ are probability density functions with support $[0,4]$ so that we can easily see that for any $t\in [0,4]$, 
$$\int_0^t f(x)dx\geq \int_0^t f_1(x)dx$$
so that the claim holds.

\begin{lemma}\label{chi-max}
Suppose that $r$ is a given number, and $\{u_j^{(n)}\}_{j=1}^n$ is a set of i.i.d. chi-squared random variables with degree of freedom $r$. Then 
$$\lim\limits_{n\rightarrow\infty}\max\limits_{j=1,\cdots,n}\frac{u_j^{(n)}}{2\log{n}}=1$$
in probability. 
\end{lemma}
\begin{proof}
Assume that $Z_r$ is a chi-squared random variable of degree of $r$, then we have 
\begin{equation}\label{gamma-tail}
P\{Z_r>x\}=\frac{\Gamma(r/2,x/2)}{\Gamma(r/2)},
\end{equation}
where $\Gamma(k,z)$ denotes the upper incomplete Gamma function. Then we use the following asymptotic behavior :
$$P\{Z_r>x\}\sim \frac{1}{\Gamma(r/2)}x^{r/2-1}e^{-x/2}
~{\rm as}~ x\rightarrow \infty.$$ For $n\rightarrow\infty$, we consider the probability
$P(\max_{1\leq j\leq n}u_j^{(n)}>2(1+\epsilon)\log{n})$. By using union bound, we see that 
\begin{eqnarray*}
&&P(\max_{1\leq j\leq n}u_j^{(n)}>2(1+\epsilon)\log{n})\\
&\leq&n\frac{1}{\Gamma(r/2)}(2(1+\epsilon)\log{n})^{r/2-1}e^{-(1+\epsilon)\log{n}}\\
&\leq&\frac{1}{\Gamma(r/2)}(2(1+\epsilon)\log{n})^{r/2-1}n^{-\epsilon}\rightarrow 0
\end{eqnarray*}
as $n\rightarrow\infty$. Now, considering the probability $P(\max_{1\leq j\leq n}u_j^{(n)}<2(1-\epsilon)\log{n})$, we see that
\begin{eqnarray*}
&&P(\max_{1\leq j\leq n}u_j^{(n)}<2(1+\epsilon)\log{n})\\
&\leq&\left(1-\frac{1}{\Gamma(r/2)}(2(1-\epsilon)\log{n})^{r/2-1}e^{-(1-\epsilon)\log{n}}\right)^n\\
&\leq&\left(1-\frac{1}{\Gamma(r/2)}(2(1-\epsilon)\log{n})^{r/2-1}\frac{1}{n^{1-\epsilon}}\right)^n\rightarrow 0
\end{eqnarray*}
as $n\rightarrow\infty$ so that the claim is proved. 
\end{proof}

\section*{Appendix B: Proof of Theorem \ref{ML-suffcond}}
The proof of Theorem \ref{ML-suffcond} basically follows the line from \cite{Reeves2008}, which provides us the information theoretic analysis for partial support recovery with maximum likelihood(ML) estimator. Let $P_e(\alpha)$ be the error probability conditioned on the true support set $S$ with fractional distortion $\alpha$. Since the sampling procedure is independent from $S$ so that for any distribution over $S$, we have $P_e(\alpha)=P_e(\alpha|S).$ Consider the sets 
\begin{eqnarray*}
G&=&\{U:|U|=k,|U\cap S|>(1-\alpha)k\},\\
B&=&\{U:|U|=k,|U\cap S|\leq (1-\alpha)k\}.
\end{eqnarray*}
Let ${\rm err}(U)=(1/\sigma_w^2)\|P_{R(A_U)}^{\perp}Y\|_F^2.$ For any $t>0$, we define two events
$$A_B=\{\min\limits_{U\in B}{\rm err}(U)<t\},~A_G=\{U:\min\limits_{V\in G}{\rm err}(V)>t\}.$$
Then $P_e(\alpha|S)\leq P(A_B)+P(A_G).$

 First, if we noting that $\min_{V\in G}{\rm err}(V)
\geq {\rm err}(K)$, we have $P(A_G)\leq P({\rm err}(K)>t)$. Since $N$ has zero mean i.i.d. Gaussian columns, and $P_{R(A_S)}^{\perp}$ is an orthogonal projection matrix with rank $m-k$, the random variable ${\rm err}(K)=(1/\sigma_w^2)\|P_{R(A_S)}^{\perp}N\|_F^2$ has a chi-squared distribution with degree of freedom $r(m-k)$ since $r$ columns of $N$ are independent.

Second, we consider $P(A_B)$. We partition $B$ by $B=\cup_{a=a_{*}}^{a^{*}}B_a$ where 
$$\tilde{B}(a)=\{U:|U|=k,|U\cap S|=k-a\},$$
$a_{*}=\lfloor \alpha k\rfloor$ and $a^{*}=\lceil (1-\epsilon)k\rceil$. Then 
$$P(A_B)\leq \sum\limits_{a=a_{*}}^{a^{*}}P(A_{\tilde{B}(a)})$$
where 
$$A_{\tilde{B}(a)}=\left\{\min\limits_{U\in \tilde{B}(a)}{\rm err}(U)<t\right\}.$$ 
Then we need to quantify the distribution of ${\rm err}(U)$ for $U\in \tilde{B}(a)$. First, if we condition on the set $S\setminus U$, the magnitude of the missed components of $X$ is given by ${\sf SNR}(X_{S\setminus U}).$ Furthermore, for any $U$, $\Lambda(U):=(1/\sigma_w^2)\|P_{R(A_U)}^{\perp}N\|_F^2$ is a chi-squared random variable with $r(m-k)$ degree of freedom by the independency of each column of $N$. Conditioned on ${\sf SNR}(X^{S\setminus U})=\theta$, the random vector 
$(\sigma_w^2\theta)^{-1/2}A_{S\setminus U}\xb^{S\setminus U}_j$ has iid zero mean Gaussian random elements with variance 1, where $\xb_j$ is the $j$-th column of $X$. If we also add an another condition $\Lambda(U)=\lambda$, then we see that 
$$\frac{1}{\theta}{\rm err}(U)=\frac{1}{\sigma_w^2\theta}\|P_{R(A_S)}^{\perp}(A_{S\setminus U}X^{S\setminus U}+N)\|_F^2$$
is a non-central chi-squared random variable with non-centrality parameter $\lambda/\theta$ and degree of freedom $r(m-k)$. This implies that 
\begin{eqnarray*}
&&P\{{\rm err}(U)<t|{\sf SNR}(X_{S\setminus U})=\theta, \Lambda(U)=\lambda\}
\\&=&P\{\chi_{NC}^2(r(m-k),\lambda/\theta)<t/\theta\}.
\end{eqnarray*}  
By the Lemma A.3 in \cite{Reeves2008}, since $\Lambda(U)\geq 0$, we have 
\begin{eqnarray*}
&&P\{{\rm err}(U)<t|{\sf SNR}(X_{S\setminus U})=\theta\}
\\&\leq&P\{\chi^2(r(m-k))<t/\theta\}
\end{eqnarray*}
using $\chi_{NC}^2(r(m-k),0)=\chi^2(r(m-k))$. Hence we have  
\begin{eqnarray*}
&&P\{{\rm err}(U)<t|{\sf SNR}(X_{S\setminus U}\geq\theta)\}
\\&\leq&P\{\chi^2(r(m-k))<t/\theta\}.
\end{eqnarray*} Then 
\begin{eqnarray*}
P(A_{\tilde{B}(a)})\leq P({\sf SNR}<\theta)+\sum\limits_{U\in \tilde{B}(a)}
P\{\chi^2(r(m-k))<t/\theta\}.
\end{eqnarray*}
By the definition of ${\sf SNR}$ and $g(a/k,X)$, 
$$\frac{{\sf SNR}(X_{S\setminus U})}{\alpha g(a/k,X)}\geq {\sf SNR}(X).$$
By the definition of $g(a/k,\mathcal{X})$ and ${\sf SNR}(\mathcal{X})$, there is a $c_0>0$ such that 
$$P\{\min\limits_{U\in\tilde{B}(a)}{\sf SNR}(X_{S\setminus U})<1/\zeta(a)\}<e^{-nc_0}$$
where $\zeta(a)=[{\sf SNR}(\mathcal{X})(a/k) g(a/k,\mathcal{X})]^{-1}$. Hence 
\begin{eqnarray*}
&&P(A_{\tilde{B}(a)})<e^{-nc_0}+\sum\limits_{U\in \tilde{B}(a)}P\{
\chi^2(r(m-k))<\zeta(a)t\}\\
&=&e^{-nc_0}+{k \choose a}{n-k \choose a}P\{\chi^2(r(m-k))<\zeta(a)t\}.
\end{eqnarray*}
Reminding $P_e(\alpha)\leq P(A_G)+P(A_B)$, we first bound $P(A_G)$.
For arbitrary $\nu>0$, we choose $t_{\nu}=(1+\nu)r(m-k)$. Then by Lemma \ref{chi-bound}, we have 
$$P\{\chi^2(r(m-k))>t_{\nu}\}\leq \exp{(-nE_1)},$$
where $E_1:=(\rho-\epsilon)\nu^2/4$. With $\nu$ arbitrary close to 0, we consider the probability $P(A_B)$. To use Lemma \ref{chi-bound}, we need the condition 
$${\sf SNR}(\mathcal{X})>\max\limits_{u\in [\alpha,1-\epsilon]}\frac{1}{ug(u,\mathcal{X})}=\frac{1}{\alpha g(\alpha,\mathcal{X})}.$$
If this condition is satisfied, then 
$$P\{\chi^2(r(m-k))<\tau(a)t\}\leq e^{-nE_2(a)}$$
where 
$$E_2(a)=\frac{\rho-\epsilon}{2}\left[-\log(\zeta(a))+\zeta(a)-1\right].$$

Finally noting that $\log{{k \choose a}{n-k \choose a}}\rightarrow nh(\epsilon,a/k)$ as $n\rightarrow\infty$, we have 
\begin{eqnarray*}
P_e(\alpha)&\leq& e^{-nE_1}+\sum\limits_{a=a_{*}}^{a^{*}}\left[
e^{-n(E_2(a)-h(\epsilon,a/k))}+e^{-nc_0}\right]\\
&<&e^{-nE_1}+\max\limits_{\alpha k\leq a\leq (1-\epsilon)k}
e^{-n(E_2(a)-h(\epsilon,a/k))+\log{k}}\\
&&+e^{-nc_0+\log{k}}.
\end{eqnarray*}
For $n\rightarrow\infty$, the ML estimator is asymptotically reliable if we have SNR condition (\ref{par-snr}) and $E_2(a)>h(\epsilon,a/k)$ for $\alpha k\leq a\leq (1-\epsilon)k$ which holds under the condition (\ref{par-sam}).

\begin{lemma}\cite{Reeves2008}\label{chi-bound}
For positive integer $r$ and random variable $Z$ which has the distribution $\chi^2(r)$ and for any $\epsilon>0$ we have 
\begin{eqnarray*}
P\{Z>(1+\epsilon)r\}&\leq& e^{-\frac{r}{4}\epsilon^2},\\
P\{Z<(1-\epsilon)r\}&\leq & \exp{\left(-\frac{r}{2}\left[
-\log{(1-\epsilon)}-\epsilon\right]\right)}.
\end{eqnarray*}
\end{lemma}

\section*{Appendix C: Proof of Theorem \ref{nec-cond}}
Again, the proof of Theorem \ref{nec-cond} is a generalization of the result for the necessary condition in \cite{Reeves2008} to the MMV cases. 

We define $Z=X_S$. Then the pair $(S,Z)$ is equivalent to $X$. By the data processing inequality and the chain rule for mutual information, we have 
\begin{eqnarray}\label{nec-1}
I(B;Y)\geq I(Z,S;Y)=I(S;Y)+I(Z;Y|S) 
\end{eqnarray}
where $B=AX$ is the noiseless measurement. Since the noise $N$ is i.i.d. Gaussian with covariance matrix $\sigma_w^2 I$ and 
$$E((AX+N)^{*}(AX+N))=X^{*}X+\sigma_w^2 I$$
for a given $X$, we can obtain an upper bound of $I(B;Y)$ as 
\begin{eqnarray}\label{nec-2}
\notag I(B;Y)&\leq& \frac{1}{2}\log{[\det{(I+\frac{1}{\sigma_w^2} X^{*}X)}]}\\
\notag &=& \sum\limits_{l=1}^r \frac{1}{2}\log{(1+\frac{1}{\sigma_w^2}\lambda_l(X^{*}X))}\\
&\leq& \sum\limits_{l=1}^r \frac{1}{2}\log{(1+\frac{1}{\sigma_w^2}\kappa_l(\mathcal{X}))}
\end{eqnarray}
asymptotically since $re^{-cn}\rightarrow 0$ for $n\rightarrow\infty$, where $\lambda_l(X^{*}X)$ is the $l$-th largest eigenvalue of $X^{*}X$ for $1\leq l\leq r$.

Then, we consider the information $I(S,Y)$. Given that $S$ is uniformly chosen over $n \choose k$ possibilities, the asymptotic number of bits we need to decode $S$ to with distortion rate $\alpha$ is given by $nh(\epsilon)-nh(\epsilon,\alpha)$, where we used 
$\log{n \choose k}=nh(k/n)+\mathcal{O}(\log{n})$. By Fano's inequality, $P_e(\alpha)=0$ only if 
\begin{equation}\label{nec-3}
I(S;Y)\geq nh(\epsilon)-nh(\epsilon,\alpha).
\end{equation}
Applying \eqref{nec-1}, \eqref{nec-2} and \eqref{nec-3} we have 
\begin{equation*}
m\geq \frac{nh(\epsilon)-nh(\epsilon,\alpha)+I(Z;Y|S)}{\sum\limits_{l=1}^r
\frac{1}{2}\log{[1+\frac{1}{\sigma_w^2}\kappa_l(\mathcal{X})]}}
\end{equation*}

\end{document}